\documentclass[11pt]{llncs}

\usepackage{graphicx,color}
\usepackage{latexsym,amsfonts}
\usepackage{amsmath}
\usepackage{amssymb}
\usepackage{bbm}
\usepackage{times}
\usepackage{mathptm}
\usepackage{algorithmic}
\usepackage{algorithm}
\usepackage{array}
\usepackage{fullpage}
\usepackage{tikz} 
\usetikzlibrary{arrows,decorations.markings,patterns,calc, positioning, backgrounds} 
\usepackage{enumerate}

\newcounter{lpnumber} 

\newcommand{\cupdot}{\mathbin{\mathaccent\cdot\cup}}

\newcommand{\vote}{\mathsf{vote}}

\newtheorem{new-claim}{Claim}

\title{Popular Edges and Dominant Matchings}
\author{\'{A}gnes Cseh\inst{1}\ \ \and \ \ Telikepalli Kavitha\inst{2}}
\institute{TU Berlin, Germany. \email{cseh@math.tu-berlin.de} \and Tata Institute of Fundamental Research, India. \email{kavitha@tcs.tifr.res.in}}

\begin{document}
\pagestyle{plain}
\maketitle

\begin{abstract}
Given a bipartite graph $G = (A \cup B,E)$ with strict preference lists and $e^* \in E$, we ask if there exists a popular matching in $G$
that contains the edge~$e^*$. We call this the {\em popular edge} problem. 
A matching $M$ is popular if there is no matching $M'$ such that the vertices that prefer $M'$ 
to $M$ outnumber those that prefer $M$ to~$M'$. 
It is known that every stable matching is popular; however $G$ may have no stable matching with the edge $e^*$ in it.
In this paper we identify another natural subclass of  popular matchings called ``dominant matchings'' and
show that if there is a popular matching that contains the edge $e^*$, then there is either
a stable matching that contains $e^*$ or a dominant matching that contains~$e^*$. 
This allows us to design a linear time algorithm for the popular edge problem. We also use dominant matchings to efficiently 
test if every popular matching in $G$ is stable or not.
\end{abstract}

\section{Introduction}
\label{intro}
We are given a bipartite graph $G = (A \cup B, E)$ where each vertex has a strict preference list ranking its neighbors and we
are also given $e^* \in E$. Our goal is to compute a matching that contains the edge $e^*$, in other words, $e^*$ is an essential
edge to be included in our matching -- however, our matching also has to be globally acceptable, in other words,
$M$ has to {\em popular} (defined below). 

We say a vertex $u \in A \cup B$ prefers matching $M$ to matching $M'$ if either $u$ is matched in $M$ and unmatched 
in $M'$ or $u$ is matched in both and in $u$'s preference list, $M(u)$, i.e., $u$'s partner in $M$, is ranked better than $M'(u)$, 
i.e., $u$'s partner in~$M'$.
For matchings $M$ and $M'$ in $G$, let $\phi(M,M')$ be the number of vertices that prefer $M$ to~$M'$.
If $\phi(M',M) > \phi(M,M')$ then we say $M'$ is {\em more popular than}~$M$.

\begin{definition}
A matching $M$ is {\em popular} if  there is no matching that is more popular than $M$; in other words,
$\phi(M,M') \ge \phi(M',M)$ for all matchings $M'$ in~$G$.
\end{definition}

Thus in an election between any pair of matchings, where each vertex casts a vote for the matching that it prefers,
a popular matching never loses. 
Popular matchings always exist in $G$ since every stable matching is popular~\cite{Gar75}. 
Recall that a matching $M$ is stable if it has no {\em blocking pair}, 
i.e., no pair $(a,b)$ such that both $a$ and $b$ prefer each other to their respective assignments in~$M$. 
It is known that every stable matching is a minimum size popular matching~\cite{HK13}, thus
the notion of popularity is a relaxation of stability.

Our problem is to determine if there exists a popular matching in $G$ that contains the essential edge~$e^*$. 
We call this the {\em popular edge} problem. 
It is easy to check if there exists a stable matching that contains~$e^*$. However as stability is stricter than popularity,
it may be the case that there is no stable matching that contains $e^*$ while there is a popular  matching that contains $e^*$.
Fig.~\ref{fig:1} has such an example.


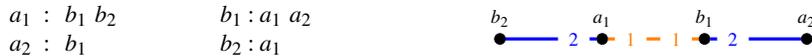
\begin{figure}[h]
	\vspace*{-4mm}
	\tikzstyle{vertex} = [circle, draw=black, fill=black, inner sep=0pt,  minimum size=5pt]
	\tikzstyle{edgelabel} = [circle, fill=white, inner sep=0pt,  minimum size=15pt]
	\centering
	\pgfmathsetmacro{\d}{1.7}
	\begin{minipage}{0.4\textwidth}
		\[
		\begin{array}{llllllll}
		a_1 \ : & & b_1 \ & b_2  & \hspace*{0.5in} b_1: a_1 \ & a_2\\
		a_2 \ : & & b_1 \ &      & \hspace*{0.5in} b_2: a_1 \ &  
		\end{array}
		\]
	\end{minipage}\hspace{10mm}\begin{minipage}{0.4\textwidth}
	\begin{tikzpicture}[scale=0.8, transform shape]
	\node[vertex, label=above:$a_1$] (a1) at (0,0) {};
	\node[vertex, label=above:$b_1$] (b1) at (\d,0) {};
	\node[vertex, label=above:$b_2$] (b3) at (-\d,0) {};
	\node[vertex, label=above:$a_2$] (a3) at ($(a1) + (2*\d, 0)$) {};
	
	\draw [very thick, orange] (a1) -- node[edgelabel, near start] {1} node[edgelabel, near end] {1} (b1);
	\draw [very thick, blue] (a1) -- node[edgelabel, near start] {2} (b3);
	\draw [very thick, blue] (a3) -- node[edgelabel, near end] {2} (b1);
	\end{tikzpicture}
\end{minipage}

\caption{The top-choice of both $a_1$ and of $a_2$ is $b_1$; the second choice of $a_1$ is~$b_2$. The preference lists of the $b_i$'s are symmetric. There is 
no edge between $a_2$ and~$b_2$. The matching $S = \{(a_1,b_1)\}$ is the only stable matching here, while there is another popular matching 
$M = \{(a_1,b_2), (a_2,b_1)\}$. Thus every edge is a popular edge here while there is only one stable edge~$(a_1,b_1)$.}
\label{fig:1}
\vspace*{-5mm}
\end{figure}

It is a theoretically interesting problem to identify those edges that can occur in a popular matching and those that cannot.
This solution is also applicable in a setting where $A$ is a set of applicants and $B$ is a set of posts, each applicant
seeks to be matched to an adjacent post and vice versa; also vertices have strict preferences over their neighbors.
The central authority desires to pair applicant $a$ and post $b$ to each other. However if the resulting matching $M$ should lose 
an election where vertices cast votes, then $M$ is unpopular and not globally stable; 
the central authority wants to avoid such a situation.
Thus what is sought here is a matching that satisfies both these conditions: (1)~$(a,b) \in M$ and (2)~$M$ is popular.

A first attempt to solve this problem may be to ask for a {\em stable} matching $S$ in the subgraph obtained by deleting the 
endpoints of $e^*$ from $G$ and add $e^*$ to~$S$. However $S \cup \{e^*\}$ need not be popular. 
Fig.~\ref{fig:arranged_pair} has a simple example where $e^* = (a_2,b_2)$ and the subgraph induced by $a_1,b_1,a_3,b_3$ 
has a unique stable matching~$\{(a_1,b_1)\}$. However $\{(a_1,b_1),(a_2,b_2)\}$ is not 
popular in $G$ as  $\{(a_1,b_3),(a_2,b_1)\}$ is more popular. 
Note that there {\em is} a popular matching $M^* = \{(a_1,b_3),(a_2,b_2),(a_3,b_1)\}$ that contains~$e^*$.

\begin{figure}[h]
\vspace*{-4mm}
\tikzstyle{vertex} = [circle, draw=black, fill=black, inner sep=0pt,  minimum size=5pt]
\tikzstyle{edgelabel} = [circle, fill=white, inner sep=0pt,  minimum size=15pt]
\centering
\pgfmathsetmacro{\d}{1.7}
\begin{minipage}{0.4\textwidth}
\[
\begin{array}{llllllll}
a_1 \ : & & b_1 \ & b_2 \ &  & \hspace*{0.5in}b_1: a_2 \ & a_1 \ & a_3\\
a_2 \ : & & b_1 \ & b_2 \ &       &   \hspace*{0.5in} b_2: a_2 \ &  \\
a_3 \ : & & b_1 \ &        &      &     \hspace*{0.5in} b_3: a_1
\end{array}
\]
\end{minipage}\hspace{10mm}\begin{minipage}{0.4\textwidth}
\begin{tikzpicture}[scale=0.8, transform shape]
\node[vertex, label=above:$a_1$] (a1) at (0,0) {};
\node[vertex, label=above:$b_1$] (b1) at ($(a1) + (\d, 0)$) {};
\node[vertex, label=left:$b_2$] (b2) at ($(a1) + (0, -\d)$) {};
\node[vertex, label=right:$a_2$] (a2) at ($(b1) + (0, -\d)$) {};
\node[vertex, label=above:$b_3$] (b3) at (-\d,0) {};
\node[vertex, label=above:$a_3$] (a3) at ($(a1) + (2*\d, 0)$) {};
\draw [very thick, blue] (a1) -- node[edgelabel, near start] {1} node[edgelabel, near end] {2} (b1);
\draw [very thick] (a1) -- node[edgelabel, near start] {2} (b3);
\draw [very thick] (a2) -- node[edgelabel, near start] {1} node[edgelabel, near end] {1} (b1);
\draw [very thick, blue] (a2) -- node[edgelabel, above] {$e^*$} node[edgelabel, near start] {2} (b2);
\draw [very thick] (a3) -- node[edgelabel, near end] {3} (b1);
\end{tikzpicture}
\end{minipage}
\caption{Here we have $e^* = (a_2,b_2)$. The top choice for $a_1$ and $a_2$ is $b_1$ while $b_2$ is their second choice; $b_1$'s top choice is $a_2$, second choice is $a_1$, and third choice is~$a_3$. The vertices $b_2,b_3$, and $a_3$ have $a_2,a_1$, and $b_1$ as their only neighbors.}
\label{fig:arranged_pair}
\vspace*{-5mm}
\end{figure}
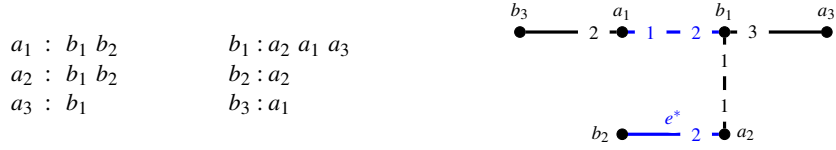

It would indeed be surprising if it was the rule that for every edge $e^*$, there is always a popular matching  
which can be decomposed as $\{e^*\}$ $\cup$ a stable matching on $E \setminus \{e^*\}$, as popularity is a far more flexible
notion than stability; for instance, the set of vertices matched in every stable matching in $G$ is the same~\cite{GS85} while
there can be a large variation (up to a factor of 2) in the sizes of popular matchings in $G$.
We need a larger palette than the set of stable 
matchings to solve the popular edge problem. We now identify another natural subclass of popular matchings called  {\em dominant} 
popular matchings or dominant matchings, in short. In order to define dominant matchings, we use the relation ``defeats'', defined 
as follows.

\begin{definition}
\label{def:defeat}
Matching $M$ {\em defeats} matching $M'$ if either of these two conditions holds:
\begin{enumerate}[(i)]
\item $M$ is more popular than $M'$, i.e., $\phi(M,M') > \phi(M',M)$;
\item $\phi(M,M') = \phi(M',M)$ and $|M| > |M'|$.
\end{enumerate}
\end{definition}

When $M$ and $M'$ gather the same number of votes in the election between $M$ and $M'$, instead of declaring these matchings as 
incomparable (as done under the ``more popular than'' relation), it seems natural to regard the larger of $M,M'$ as the {\em winner} 
of the election. Condition~(ii) of the {\em defeats} relation exactly captures this. We define dominant matchings to be those popular 
matchings that are never defeated (as per Definition~\ref{def:defeat}).

\begin{definition}
\label{def:dominant}
Matching $M$ is {\em dominant} if there is no matching that {\em defeats} it; in other words, $M$ is popular and for any matching 
$M'$, if $|M'| > |M|$, then $M$ is more popular than~$M'$. 
\end{definition}

Note that a dominant matching has to be a maximum size popular matching since smaller-sized popular matchings get defeated by a 
popular matching of maximum size.  However not every maximum size popular matching is a dominant matching, as the example (from~\cite{HK13}) in Fig.~\ref{fig:0} demonstrates.

\begin{figure}[h]
\vspace*{-4mm}
\tikzstyle{vertex} = [circle, draw=black, fill=black, inner sep=0pt,  minimum size=5pt]
\tikzstyle{edgelabel} = [circle, fill=white, inner sep=0pt,  minimum size=15pt]
\centering
	\pgfmathsetmacro{\d}{1.7}
\begin{minipage}{0.4\textwidth}
\[
\begin{array}{llllllll}
a_1 \ : & & b_1 \ & b_2 \ & b_3 \ & \hspace*{0.5in}b_1: a_1 \ & a_2 \ & a_3\\
a_2 \ : & & b_1 \ & b_2 \ &       &   \hspace*{0.5in} b_2: a_1 \ & a_2\  \\
a_3 \ : & & b_1 \ &     \  &      &     \hspace*{0.5in} b_3: a_1
\end{array}
\]
\end{minipage}\hspace{10mm}\begin{minipage}{0.4\textwidth}
\begin{tikzpicture}[scale=0.8, transform shape]
	\node[vertex, label=above:$a_1$] (a1) at (0,0) {};
	\node[vertex, label=left:$b_2$] (b1) at ($(a1) + (0, -\d)$) {};
	\node[vertex, label=above:$b_1$] (b2) at (\d,0) {};
	\node[vertex, label=right:$a_2$] (a2) at ($(b2) + (0, -\d)$) {};
	\node[vertex, label=above:$b_3$] (b3) at (-\d,0) {};
	\node[vertex, label=above:$a_3$] (a3) at ($(a1) + (2*\d, 0)$) {};
    
	\draw [very thick, blue] (a1) -- node[edgelabel, near start] {2} node[edgelabel, near end] {1} (b1);
	\draw [very thick, orange] (a1) -- node[edgelabel, near start] {1} node[edgelabel, near end] {1} (b2);
    \draw [very thick] (a1) -- node[edgelabel, near start] {3} (b3);
	\draw [very thick, orange] (a2) -- node[edgelabel, near start] {2} node[edgelabel, near end] {2} (b1);
	\draw [very thick, blue] (a2) -- node[edgelabel, near start] {1} node[edgelabel, near end] {2} (b2);
    \draw [very thick] (a3) -- node[edgelabel, near end] {3} (b2);
\end{tikzpicture}
\end{minipage}

\caption{The vertex $b_1$ is the top choice for all $a_i$'s and $b_2$ is the second choice for $a_1$ and $a_2$ while $b_3$ is the third choice for~$a_1$. The preference lists of the $b_i$'s are symmetric. There are 2 maximum size popular matchings here: $M_1 = \{(a_1,b_1),(a_2,b_2)\}$ and $M_2 = \{(a_1,b_2),(a_2,b_1)\}$. The matching $M_1$ is not dominant since the larger matching $M_3 = \{(a_1,b_3),(a_2,b_2),(a_3,b_1)\}$ defeats it. The matching $M_2$ is dominant since $M_2$ is more popular than~$M_3$.}
\label{fig:0}
\vspace*{-5mm}
\end{figure}
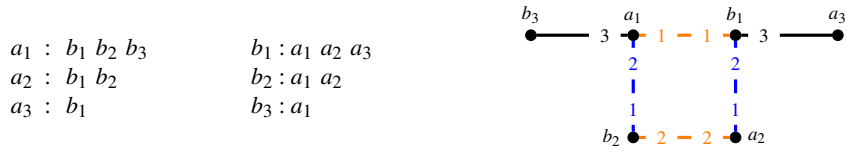



Analogous to Definition~\ref{def:dominant}, we can define the following subclass of popular matchings: 
those popular matchings $M$ such that {\em for any matching $M'$, if $|M'| < |M|$ then $M$ is more popular than $M'$}. 
It is easy to show that this class of matchings is 
exactly the set of stable matchings. That is, we can show that any popular matching $M$ that 
is more popular than every {\em smaller-sized} matching is a stable matching and conversely, every
stable matching is more popular than any smaller-sized matching. 
Thus dominant matchings are to the class of maximum size popular matchings what stable
matchings are to the class of minimum size popular matchings: these are popular matchings that carry 
the proof of their maximality (similarly, minimality) by
being more popular than every matching of larger (resp., smaller) size.

\subsubsection{Our contribution.} 
Theorem~\ref{main-thm} is our main result here. This enables us to solve the popular edge problem in linear time.
\begin{theorem}
\label{main-thm}
If there exists a popular matching in $G = (A\cup B,E)$ that contains the edge $e^*$, then there exists 
either a stable matching in $G$ that contains $e^*$ or a dominant matching in $G$ that contains~$e^*$.
\end{theorem}

\begin{enumerate}[(i)]
\item To show Theorem~\ref{main-thm}, we show that any popular matching $M$ can be partitioned as $M_0 \cupdot M_1$, where $M_0$ is dominant in the
subgraph induced by the vertices matched in $M_0$ and in the subgraph induced by the remaining vertices, $M_1$ is stable. If $M$ contains $e^*$, 
then $e^*$ is either in $M_0$ or in~$M_1$. In the former case, we show
a dominant matching in $G$ that contains $e^*$ and in the latter case, we show a stable matching in $G$ that contains~$e^*$.

\item We also show that every dominant matching in $G$ can be realized as an image (under a simple and natural mapping) of a stable matching in a new graph~
$G'$. This allows us to determine in linear time if there is a dominant matching in $G$ that contains the edge~$e^*$.
This mapping between stable matchings in $G'$ and dominant matchings in $G$ can also be used to find a min-cost dominant matching in~$G$ 
efficiently, where we assume there is a rational cost function on~$E$.

\item When all popular matchings in $G$ have the same size, it could be the case that every popular matching in $G$ is also stable.
That is, in $G$ we have $\{$popular matchings$\} = \{$stable matchings$\}$. We use dominant matchings to efficiently check if
this is the case or not. 
We show that if there exists an unstable popular matching in $G$, then there has to exist an 
unstable dominant matching in $G$. This allows us to design an $O(m^2)$ algorithm (where $|E| = m)$ to check if every 
popular matching in $G$ is also stable. 
\end{enumerate}

\subsubsection{Related Results.} Stable matchings were defined by Gale and Shapley in their landmark paper~\cite{GS62}. The attention of the community was drawn very early to the characterization of \emph{stable edges}: edges and sets of edges that can appear in a stable matching. In the seminal book of Knuth~\cite{Knu76}, stable edges first appeared under the term ``arranged marriages''. Knuth presented an algorithm to find a stable matching with a given stable set of edges or report that none exists. This method is a modified version of the Gale-Shapley algorithm and runs in $O(m)$ time. 
Gusfield and Irving~\cite{GI89} provided a similar, simple method for the stable edge problem with the same running time. 

The stable edge problem is a highly restricted case of the \emph{min-cost stable matching problem}, where a stable matching that has the 
minimum edge cost among all stable matchings is sought. With the help of edge costs, various stable matching problems can be modeled, such 
as stable matchings with restricted edges~\cite{DFFS03} or egalitarian stable matchings~\cite{ILG87}.
A simple and elegant formulation of the stable matching polytope of $G = (A \cup B,E)$ is known~\cite{Rot92} and using this, a min-cost stable 
matching can be computed in polynomial time via linear programming. 

The size of a stable matching in $G$ can be as small as $|M_{\max}|/2$, where $M_{\max}$ is a maximum size matching in $G$.
Relaxing stability to popularity yields larger matchings and it is easy to show that a largest popular matching has
size at least~$2|M_{\max}|/3$. Efficient algorithms for computing a popular matching of maximum size were 
shown in~\cite{HK13,Kav12-journal}. The algorithm in \cite{HK13} time runs in $O(mn_0)$ time, where $n_0 = \min(|A|,|B|)$ and the algorithm in 
\cite{Kav12-journal} runs in linear time. In fact, both these algorithms compute dominant matchings -- thus 
dominant matchings always exist in a stable marriage instance with strict preference lists. Interestingly,
all the polynomial time algorithms currently known for 
computing {\em any} popular matching in $G=(A\cup B,E)$ compute either a stable matching or a dominant matching in~$G$. 

\paragraph{Organization of the paper.}
A characterization of dominant matchings is given in Section~\ref{sec:char}. In Section~\ref{sec:dom-mat} we show a surjective mapping between
stable matchings in a larger graph $G'$ and dominant matchings in~$G$. Section~\ref{sec:pop-edge} has our algorithm for the popular edge problem and
Section~\ref{sec:dom-vs-stab} has our algorithm to test if every popular matching in $G$ is also stable.
The Appendix has a brief overview  of the  maximum size popular matching algorithms in \cite{HK13,Kav12-journal}.

\section{A characterization of dominant matchings}
\label{sec:char}

Let $M$ be any matching in $G = (A \cup B, E)$ and let $M(u)$ denote $u$'s partner in $M$, where $u \in A \cup B$. Label each edge $e=(a,b)$ in $E\setminus M$ by the pair $(\alpha_e,\beta_e)$,
where $\alpha_e = \vote_a(b,M(a))$  and $\beta_e =  \vote_b(a,M(b))$, i.e.,  $\alpha_e$ is $a$'s vote for $b$ vs.\ $M(a)$ and $\beta_e$ is
$b$'s vote for $a$ vs.~$M(b)$.  The function $\vote(\cdot,\cdot)$ is defined below.
\begin{definition}
For any $u \in A\cup B$ and neighbors $x$ and $y$ of $u$, define $u$'s
vote between $x$ and $y$ as:
\begin{equation*} 
\vspace*{-2mm}
\label{vote-defn}
\vote_u(x,y) = \begin{cases} +   & \text{if  $u$ prefers $x$ to $y$}\\
	                     - &  \text{if  $u$ prefers $y$ to $x$}\\			
                              0 & \text{otherwise (i.e., $x = y$).}
\end{cases}
\end{equation*}
\end{definition}
If a vertex $u$ is unmatched, then $M(u)$ is undefined and we define $\vote_u(v,M(u))$ to be $+$ for all 
neighbors of $u$ since every vertex prefers to be matched than be unmatched.
Note that if an edge $(a,b)$ is labeled $(+,+)$, then $(a,b)$ blocks $M$ in the stable matching sense.
If an edge $(a,b)$ is labeled $(-,-)$, then both $a$ and $b$ prefer their respective partners in $M$
to each other. Let $G_M$ be the subgraph of $G$ obtained by deleting edges that are labeled~$(-,-)$.
The following theorem  characterizes popular matchings.

\begin{theorem}[from \cite{HK13}]
\label{thm:pop-char}
A matching $M$ is popular if and only if the following three conditions are satisfied in the subgraph $G_M$:
\begin{enumerate}
\item[(i)]There is no alternating cycle with respect to $M$ that contains a $(+,+)$ edge.
\item[(ii)]There is no alternating path starting from an unmatched vertex wrt $M$ that contains a $(+,+)$ edge.
\item[(iii)]There is no alternating path with respect to $M$ that contains two or more $(+,+)$ edges.
\end{enumerate}
\end{theorem}

Lemma~\ref{thm:domn-char} characterizes those popular matchings that are dominant. 
The ``if'' side of Lemma~\ref{thm:domn-char} was shown in \cite{Kav12-journal}:
it was shown that if there is no augmenting path with respect to a popular matching $M$ in $G_M$ then $M$ is more popular than all larger 
matchings, thus $M$ is a maximum size popular matching. 

Here we show that the converse holds as well, i.e., if $M$ is a popular matching such that $M$ is more popular than all larger matchings,
in other words, if $M$ is a dominant matching,
then there is no augmenting path with respect to $M$ in~$G_M$.
\begin{lemma}
\label{thm:domn-char}
A popular matching $M$ is dominant if and only if there is no augmenting path wrt $M$ in~$G_M$.
\end{lemma}
\begin{proof}
Let $M$ be a popular matching in~$G$. Suppose there is an augmenting path $\rho$ with respect to  $M$ in~$G_M$.
Let us use $M \approx M'$ to denote both matchings getting the same number of votes in an election between them, i.e., $\phi(M,M') = \phi(M',M)$.
We will now show that $M\oplus\rho \approx M$. Since $M\oplus\rho$ is a larger matching than $M$, if $M\oplus\rho \approx M$, 
then it means that $M\oplus\rho$ defeats $M$, thus $M$ is not dominant.

Consider $M\oplus\rho$ versus $M$: every vertex that does not belong to the path $\rho$ gets the same partner in
both these matchings. Hence vertices outside $\rho$ are indifferent between these two matchings. Consider the vertices on~$\rho$.
In the first place, there is no edge in $\rho\setminus M$ that is labeled $(+,+)$, otherwise that would contradict
condition~(ii) of Theorem~\ref{thm:pop-char}. Since the path $\rho$ belongs to $G_M$, no edge is labeled $(-,-)$ either.
Hence every edge in $\rho\setminus M$ is labeled either $(+,-)$ or~$(-,+)$. Note that the $+$ signs count the number of votes for
$M\oplus\rho$ while the $-$ signs count the number of votes for~$M$. Thus the number of votes for $M\oplus\rho$ equals the
number of votes for $M$ on vertices of $\rho$, and thus in the entire graph~$G$. Hence $M\oplus\rho \approx M$.

Now we show the other direction: if there is no augmenting path with respect to a popular matching $M$ in $G_M$ then $M$ is dominant.
Let $M'$ be a larger matching. Consider $M \oplus M'$ in $G$:
this is a collection of alternating paths and alternating cycles and since $|M'| > |M|$, there is at least one augmenting path with respect to
$M$ here. Call this path $p$, running from vertex $u$ to vertex~$v$. Let us count the number of votes for $M$ versus $M'$ among the vertices of~$p$. 

\begin{figure}[h]
\centerline{\resizebox{0.8\textwidth}{!}{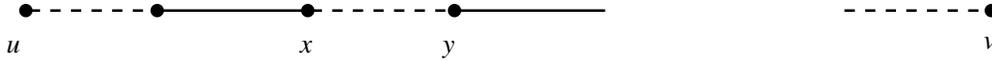}}
\caption{The $u$-$v$ augmenting path $p$ in $G$ where the bold edges are in $M$; at least one edge here (say, $(x,y)$) is labeled~$(-,-)$.}
\label{fig:thm2}
\end{figure}

No edge in $p$ is labeled $(+,+)$ as that would contradict condition~(ii) of Theorem~\ref{thm:pop-char}, thus
all the edges of $M'$ in $p$ are labeled $(-,+)$, $(+,-)$, or~$(-,-)$. Since $p$ does not exist in $G_M$, there is at least one edge that 
is labeled $(-,-)$ here (see Fig.~\ref{fig:thm2}): thus among the vertices of $p$, $M$ gets more votes than $M'$ (recall that $+$'s are votes for 
$M'$ and $-$'s are votes for $M$). Thus $M$ is more popular than $M'$ among the vertices of~$p$.

By the popularity of $M$, we know that $M$ gets at least as many votes as $M'$ over all other paths and cycles in $M \oplus M'$; this is because
if $\rho$ is an alternating path/cycle in $M \oplus M'$ such that the number of vertices on $\rho$ that prefer $M'$ to $M$ is more than 
the number that prefer $M$ to $M'$, then $M\oplus\rho$ is more popular than $M$, a contradiction to the popularity of~$M$.
Thus adding up over all the vertices in $G$, it follows that $\phi(M,M') > \phi(M',M)$. 
Hence $M$ is more popular than any larger matching and so $M$ is a dominant matching. \qed
\end{proof}

Corollary~\ref{cor0} is a characterization of dominant matchings. This follows immediately from Lemma~\ref{thm:domn-char} and
Theorem~\ref{thm:pop-char}.
\begin{corollary}
\label{cor0}
Matching $M$ is a dominant matching if and only if $M$ satisfies conditions~(i)-(iii) of Theorem~\ref{thm:pop-char} and 
condition~(iv): there is no augmenting path wrt $M$ in~$G_M$.
\end{corollary}


\section{The set of dominant matchings}
\label{sec:dom-mat}

In this section we show a surjective mapping from the set of stable matchings 
in a new instance $G' = (A'\cup B', E')$ to the set of dominant matchings in $G = (A\cup B,E)$. 
It will be convenient to refer to vertices in $A$ and $A'$ as {\em men} and vertices in $B$ and $B'$ as {\em women}. 
The construction of $G' = (A'\cup B', E')$ is as follows. 

Corresponding to every man $a \in A$, there will be
two men $a_0$ and $a_1$ in $A'$ and one woman $d(a)$ in~$B'$. The vertex $d(a)$ will be referred to as the dummy woman
corresponding to~$a$. Corresponding to every woman $b \in B$, there will be exactly one woman in $B'$ -- for the sake of simplicity, we will use 
$b$ to refer to this woman as well. Thus $B' = B \cup d(A)$, where $d(A) = \{d(a): a \in A\}$ is the set of dummy women. 

Regarding the other side of the graph, $A' = A_0 \cup A_1$, where $A_i = \{a_i: a \in A\}$ for $i = 0,1$, and vertices in $A_0$ are called level~0 
vertices, while vertices in $A_1$ are called level~1 vertices. 

We now describe the edge set $E'$ of~$G'$. 
For each $a \in A$, the vertex $d(a)$ has exactly two neighbors: these are $a_0$ and $a_1$ and
$d(a)$'s preference order is $a_0$ followed by~$a_1$. 
The dummy woman $d(a)$ is $a_1$'s most preferred neighbor and $a_0$'s least preferred neighbor.
\begin{itemize}
\item The preference list of $a_0$ is all the neighbors of $a$ (in $a$'s preference order) followed by~$d(a)$.
\item The preference list of $a_1$ is $d(a)$ followed by the neighbors of $a$ (in $a$'s preference order) in~$G$.
\end{itemize}

For any $b \in B$, its preference list in $G'$ is level~1 neighbors in the same order of preference as in $G$ followed by
level~0 neighbors in the same order of preference as in~$G$. 
For instance,  if $b$'s preference list in $G$ is $a$ followed by $a'$, then $b$'s  preference list in $G'$ is top-choice
$a_1$, then $a'_1$, and then $a_0$, and the last-choice is $a'_0$. We show an example in Fig.~\ref{fig:new-graph}.

\begin{center}
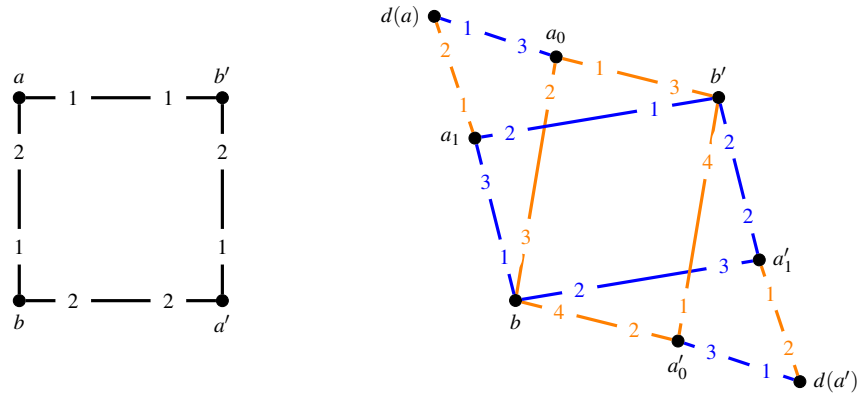
\begin{figure}[h]
\tikzstyle{vertex} = [circle, draw=black, fill=black, inner sep=0pt,  minimum size=5pt]
\tikzstyle{edgelabel} = [circle, fill=white, inner sep=0pt,  minimum size=15pt]
\centering
	\pgfmathsetmacro{\d}{3}
	\pgfmathsetmacro{\b}{4}
	\pgfmathsetmacro{\e}{0.6}
\begin{minipage}{0.3\textwidth}
\begin{tikzpicture}[scale=0.9, transform shape]
	\node[vertex, label=above:$a$] (a1) at (0,0) {};
	\node[vertex, label=below:$b$] (b1) at ($(a1) + (0, -\d)$) {};
	\node[vertex, label=above:$b'$] (b2) at (\d,0) {};
	\node[vertex, label=below:$a'$] (a2) at ($(b2) + (0, -\d)$) {};

	\draw [very thick] (a1) -- node[edgelabel, near start] {2} node[edgelabel, near end] {1} (b1);
	\draw [very thick] (a1) -- node[edgelabel, near start] {1} node[edgelabel, near end] {1} (b2);
	\draw [very thick] (a2) -- node[edgelabel, near start] {2} node[edgelabel, near end] {2} (b1);
	\draw [very thick] (a2) -- node[edgelabel, near start] {1} node[edgelabel, near end] {2} (b2);
\end{tikzpicture}
\end{minipage}\begin{minipage}{0.3\textwidth}
\begin{tikzpicture}[scale=0.9, transform shape]
	\node[vertex, label=left:$a_1$] (a1b) at ($(0,0)-(\e, \e)$) {};
	\node[vertex, label=above:$a_0$] (a1r) at ($(0,0)+(\e, \e)$) {};
	\node[vertex, label=left:$d(a)$] (ba1) at ($(0,0)+(-2*\e, 2*\e)$) {};
	
	\node[vertex, label=below:$a'_0$] (a2r) at ($(\d,-\d)-(\e, \e)$) {};
	\node[vertex, label=right:$a'_1$] (a2b) at ($(\d,-\d)+(\e, \e)$) {};
	\node[vertex, label=right:$d(a')$] (ba2) at ($(\d,-\d)+(2*\e, -2*\e)$) {};
	
	\node[vertex, label=below:$b$] (b1) at ($(a1) + (0, -\d)$) {};
	\node[vertex, label=above:$b'$] (b2) at (\d,0) {};

	\draw [very thick, orange] (a1r) -- node[edgelabel, very near start] {2} node[edgelabel, near end] {3} (b1);
	\draw [very thick, blue] (a1b) -- node[edgelabel, near start] {3} node[edgelabel, near end] {1} (b1);
	\draw [very thick, orange] (a2r) -- node[edgelabel, near start] {2} node[edgelabel, near end] {4} (b1);
	\draw [very thick, blue] (a2b) -- node[edgelabel, very near start] {3} node[edgelabel, near end] {2} (b1);
	\draw [very thick, orange] (a1r) -- node[edgelabel, near start] {1} node[edgelabel, near end] {3} (b2);
	\draw [very thick, blue] (a1b) -- node[edgelabel, very near start] {2} node[edgelabel, near end] {1} (b2);
	\draw [very thick, orange] (a2r) -- node[edgelabel, very near start] {1} node[edgelabel, near end] {4} (b2);
	\draw [very thick, blue] (a2b) -- node[edgelabel, near start] {2} node[edgelabel, near end] {2} (b2);
	
	\draw [very thick, blue] (a1r) -- node[edgelabel, near start] {3} node[edgelabel, near end] {1} (ba1);
	\draw [very thick, orange] (a1b) -- node[edgelabel, near start] {1} node[edgelabel, near end] {2} (ba1);
	\draw [very thick, blue] (a2r) -- node[edgelabel, near start] {3} node[edgelabel, near end] {1} (ba2);
	\draw [very thick, orange] (a2b) -- node[edgelabel, near start] {1} node[edgelabel, near end] {2} (ba2);
	\end{tikzpicture}
\end{minipage}
\caption{The graph $G'$ on the right corresponding to $G$ on the left. We used blue to color edges in $(A_1 \times B) \cup (A_0 \times d(a))$ and orange to color edges in $(A_0 \times B) \cup (A_1 \times d(a))$.}
\label{fig:new-graph}
\end{figure}
\end{center}

\vspace*{-0.5in}

We now define the mapping $T: \{$stable matchings in $G'\} \rightarrow \{$dominant matchings in~$G\}$.
Let $M'$ be any stable matching in~$G$.
\begin{itemize}
\item $T(M')$ is the set of edges obtained by deleting all edges involving vertices in $d(A)$ (i.e., dummy women)
from $M'$ and replacing every edge $(a_i,b) \in M'$, where $b \in B$ and $i \in \{0,1\}$, by the edge~$(a,b)$.
\end{itemize}

It is easy to see that $T(M')$ is a valid matching in~$G$. This is because $M'$ has to match $d(a)$, for every $a \in A$, 
since $d(a)$ is the top-choice for~$a_1$. Thus for each $a \in A$, one of $a_0, a_1$ has to be matched to~$d(a)$. Hence 
at most one of $a_0,a_1$ is matched to a non-dummy woman $b$ and thus $M = T(M')$ is a matching in~$G$.

\paragraph{The proof that $M$ is a dominant matching in $G$.}
This proof is similar to the proof of correctness of the maximum size popular 
matching algorithm in \cite{Kav12-journal}. As described in Section~\ref{sec:char}, in the graph $G$, label 
each edge $e = (a,b)$ in $E \setminus M$ by the pair $(\alpha_e,\beta_e)$, where $\alpha_e \in \{+,-\}$ is 
$a$'s vote for $b$ vs.\ $M(a)$ and $\beta_e \in \{+,-\}$ is $b$'s vote for $a$ vs.~$M(b)$. 

\begin{itemize}
\item It will be useful to assign a value in $\{0,1\}$ to each $a \in A$. 
If $M'(a_1) = d(a)$, then $f(a) = 0$ else $f(a) = 1$. So if $a \in A$ is unmatched in $M$ then $(a_0,d(a)) \in M'$ and so $f(a) = 1$. 
\item We will now define $f$-values for vertices in $B$ as well.
If $M'(b) \in A_1$ then $f(b) = 1$, else $f(b) = 0$.
So if $b \in B$ is unmatched in $M'$ (and thus in $M$) then $f(b) = 0$. 
\end{itemize}

\begin{new-claim}
\label{clm1}
The following statements hold on the edge labels:
\begin{itemize}
\item[(1)] If the edge $(a,b)$ is labeled $(+,+)$, then $f(a) = 0$ and $f(b) = 1$. 
\item[(2)] If $(y,z)$ is an edge such that $f(y) = 1$ and $f(z) = 0$, then $(y,z)$ has to be labeled~$(-,-)$.
\end{itemize}
\end{new-claim}
\begin{proof}
We show part~(1) first. The edge $(a,b)$ is labeled~$(+,+)$. Let $M(a) = z$ and $M(b) = y$. Thus in $a$'s preference list, 
$b$ ranks better than $z$ and similarly, in $b$'s preference list, $a$ ranks better than~$y$. We know from the definition of our 
function $T$ that $M'(z) \in \{a_0,a_1\}$ and  $M'(b) \in \{y_0,y_1\}$.
So there are 4 possibilities: $M'$ contains (1)~$(a_0,z)$ and $(y_0,b)$,
(2)~$(a_1,z)$ and $(y_0,b)$, (3)~$(a_1,z)$ and $(y_1,b)$, (4)~$(a_0,z)$ and~$(y_1,b)$.

We know that $M'$ has no blocking pairs in $G'$ since it is a stable matching.
In (1), the pair $(a_0,b)$ blocks $M'$, and in (2) and (3), the pair $(a_1,b)$ blocks~$M'$. 
Thus the only possibility is (4). That is, $M'(b) \in A_1$ and $M'(a_1) = d(a)$. In other words,
$f(a) = 0$ and $f(b) = 1$.

We now show part~(2) of Claim~\ref{clm1}. We are given that $f(y) = 1$, so $M'(y_0) = d(y)$.
We know that $d(y)$ is $y_0$'s last choice and $y_0$ is adjacent to $z$, thus $y_0$ must have been rejected by~$z$. 
Since we are given that $f(z) = 0$, i.e., $M'(z) \in A_0$, it follows that $M'(z) = u_0$, where $u$ ranks better than $y$ 
in $z$'s preference list in~$G$. 

In the graph $G'$, the vertex $z$ prefers $y_1$ to $u_0$ since it prefers any level~1 neighbor 
to a level~0 neighbor. Thus $y_1$ is matched to a neighbor that is ranked better than $z$ in $y$'s preference list, i.e., $M'(y_1) = v$, 
where $y$ prefers $v$ to~$z$. We have the edges $(y,v)$ and $(u,z)$ in $M$, thus both $y$ and $z$ prefer their respective partners in $M$ 
to each other. Hence the edge $(y,z)$ has to be labeled~$(-,-)$. \qed
\end{proof}

Lemmas~\ref{lem:aug-path} and \ref{lem:popular} shown below, along with Lemma~\ref{thm:domn-char}, 
imply that $M$ is a dominant matching in~$G$.

\begin{lemma}
\label{lem:aug-path}
There is no augmenting path with respect to $M$ in~$G_M$.
\end{lemma}
\begin{proof}
Let $a \in A$ and $b \in B$ be unmatched in~$M$. Then $f(a) = 1$ and $f(b) = 0$. 
If there is an augmenting path $\rho = \langle a, \cdots, b\rangle$ with respect to $M$ in $G_M$, then in $\rho$ we move from 
a man whose $f$-value is 1 to a woman whose $f$-value is 0. Thus there have to be two consecutive vertices $y \in A$ and $z \in B$ on 
$\rho$ such that $f(y) = 1$ and $f(z) = 0$. However part~(2) of Claim~\ref{clm1} tells us that such an edge $(y,z)$ has to be labeled~$(-,-)$. In other words, $G_M$ does not contain the edge $(y,z)$ or equivalently, there is no augmenting path $\rho$ in~$G_M$. \qed
\end{proof}

\begin{lemma}
\label{lem:popular}
$M$ is a popular matching in~$G$.
\end{lemma}
\begin{proof}
We will show that $M$ satisfies conditions~(i)-(iii) of Theorem~\ref{thm:pop-char}. 

\smallskip

\noindent{\em Condition~(i).} Consider any alternating cycle $C$ with respect to $M$ in $G_M$ and let $a$ be any vertex in $C$: if 
$f(a) = 0$ then its partner $b = M(a)$ also satisfies $f(b) = 0$ and part~(2) of Claim~\ref{clm1} tells us that there is no edge in 
$G_M$ between $b$ and any $a'$ such that $f(a') = 1$. Similarly, if $f(a) = 1$ then its partner $b = M(a)$ also satisfies $f(b) = 1$ and though
there can be an edge $(y,b)$ labeled $(+,+)$ incident on $b$, part~(1) of Claim~\ref{clm1} tells us that $f(y) = 0$ and thus there is no 
way the cycle $C$ can return to $a$, whose $f$-value is 1.
Hence if $G_M$ contains an alternating cycle $C$ with respect to $M$, then all vertices in $C$ have the same $f$-value. Since there can be no edge
labeled $(+,+)$ between 2 vertices whose $f$-value is the same (by part~(1) of Claim~\ref{clm1}), it follows that $C$ has no edge labeled $(+,+)$.

\smallskip

\noindent{\em Condition~(ii).} Consider any alternating path $p$ with respect to $M$ in $G_M$ and let the starting vertex in $p$ be~$a \in A$. Since $a$ is unmatched in $M$, we have $f(a) = 1$ and we know from part~(2) of Claim~\ref{clm1} that there is no edge in $G_M$ between such
a man and a woman whose $f$-value is 0. Thus $a$'s neighbor is $p$ is a woman $b'$ such that $f(b')=1$. Since $f(b') = 1$, its partner 
$a' = M(b')$ also satisfies $f(a') = 1$ and part~(2) of Claim~\ref{clm1} tells us that there is no edge in $G_M$ between $a'$ and any $b''$ such that 
$f(b'') = 0$, thus all vertices of $p$ have $f$-value 1 and thus there is no  edge labeled $(+,+)$ in~$p$. 

Suppose the  starting vertex in $p$ is~$b \in B$. Since $b$ is unmatched in $M$, we have $f(b) = 0$ and we again know from part~(2) of Claim~\ref{clm1} 
that there is no edge in $G_M$ between such a woman and a man whose $f$-value is 1. Thus $b$'s neighbor in $p$ is a woman $a'$ such that~$f(a')=0$. 
Since $f(a') = 0$, its partner 
$b' = M(a')$ also satisfies $f(b') = 0$ and part~(2) of Claim~\ref{clm1} tells us that there is no edge in $G_M$ between $b'$ and any $a''$ such that 
$f(a'') = 1$, thus all vertices of $p$ have $f$-value 0 and thus there is no edge labeled $(+,+)$ in~$p$. 

\smallskip

\noindent{\em Condition~(iii).} Consider any alternating path $\rho$ with respect to $M$ in~$G_M$. We can assume that the starting vertex 
in $\rho$ is matched in $M$ (as condition~(ii) has dealt with the case when this vertex is unmatched). Suppose the starting vertex is~$a \in A$. 
If $f(a) = 0$ then its partner $b = M(a)$ also satisfies $f(b) = 0$ and part~(2) of Claim~\ref{clm1} tells us that there is no edge in $G_M$ between $b$ 
and any $a'$ such that $f(a') = 1$, thus all vertices of $\rho$ have $f$-value 0 and thus there is no edge labeled $(+,+)$
in~$\rho$.  If $f(a) = 1$ then after traversing some vertices whose $f$-value is 1, we can encounter an edge $(y,z)$ that is labeled $(+,+)$ where
$f(z) = 1$ and $f(y) = 0$. However once we reach $y$, we get stuck in vertices whose $f$-value is 0 and thus we can see 
no more edges labeled $(+,+)$.

Suppose the starting vertex in $\rho$ is $b \in B$. If $f(b) = 1$ then its partner $a = M(b)$ also satisfies $f(a) = 1$ and part~(2) of 
Claim~\ref{clm1} tells us that there is no edge in $G_M$ between $a$ and any $b'$ such that $f(b') = 0$, thus all vertices of $\rho$ have $f$-value 
1 and thus there is no edge labeled $(+,+)$ in~$\rho$.  If $f(b) = 0$ then after traversing some vertices whose $f$-value is 0, we can encounter an 
edge $(y,z)$ labeled $(+,+)$ where $f(y) = 0$ and $f(z) = 1$. However once we reach $z$, we get stuck in vertices whose $f$-value is 1 and 
thus we can see no more edges labeled $(+,+)$. Thus in all cases there is at most one edge labeled $(+,+)$ in~$\rho$. \qed
\end{proof}

\subsection{$T$ is surjective}
\label{sec:onto}
We now show that corresponding to any dominant matching $M$ in $G$, there is a stable matching $M'$ in $G'$ such that $T(M') = M$.
Given a dominant matching $M$ in $G$, we first label each edge $e=(a,b)$ in $E \setminus M$ by the pair $(\alpha_e,\beta_e)$ where 
$\alpha_e$ is $a$'s vote for $b$ vs.\ $M(a)$ and $\beta_e$ is $b$'s vote for $a$ vs.~$M(b)$. We will work in $G_M$, the
subgraph of $G$ obtained by deleting all edges labeled $(-,-)$. 
We now construct sets $A_0,A_1 \subseteq A$ and $B_0, B_1 \subseteq B$ as described in the algorithm below. 
These sets will be useful in constructing the matching~$M'$.

\begin{enumerate}
\item[0.] Initialize $A_0 = B_1 = \emptyset$, \ $A_1 = \{$unmatched men in $M\}$, and $B_0 = \{$unmatched women in $M\}$.

\item[1.] For every edge $(y,z) \in M$ that is labeled $(+,+)$ do: 
\begin{itemize}
\item let $A_0 = A_0 \cup \{y\}$, \ $B_0 = B_0 \cup \{M(y)\}$, \ $B_1 = B_1 \cup \{z\}$, \ and $A_1 = A_1 \cup \{M(z)\}$.
\end{itemize}
\item[2.] While there exists a matched man $a \notin A_0$ that is adjacent in $G_M$ to a woman in $B_0$ do: 
\begin{itemize}
\item $A_0 = A_0 \cup \{a\}$ and $B_0 = B_0 \cup \{M(a)\}$.
\end{itemize}
\item[3.] While there exists a matched woman $b \notin B_1$ that is adjacent in $G_M$ to a man in $A_1$ do: 
\begin{itemize}
\item $B_1 = B_1 \cup \{b\}$ and $A_1 = A_1 \cup \{M(b)\}$.
\end{itemize}
\end{enumerate}

All unmatched men are in $A_1$ and all unmatched women are in~$B_0$. For every edge $(y,z)$ that is labeled $(+,+)$, we add
$y$ and its partner to $A_0$ and $B_0$, respectively while $z$ and its partner are added to $B_1$ and $A_1$, respectively.
For any man $a$, if $a$ is adjacent to a vertex in $B_0$ and $a$ is not in $A_0$, then $a$ and its partner get added to
$A_0$ and $B_0$, respectively. Similarly, for any woman $b$, if $b$ is adjacent to a vertex in $A_1$ and $b$ is not in $B_1$, 
then $b$ and its partner get added to $B_1$ and $A_1$, respectively.

The following observations are easy to see (refer to Fig.~\ref{fig:lem4}). Every $a \in A_1$ has an even length alternating path in $G_M$ to either 
\begin{itemize}
\item[(1)] a man unmatched in $M$ (by Step~0 and Step~3) or 
\item[(2)] a man $M(z)$ where $z$ has an edge labeled $(+,+)$ incident on it (by Step~1 and Step~3). 
\end{itemize}

Similarly, every $a \in A_0$ has an odd length alternating path in $G_M$ to either 
\begin{itemize}
\item[(3)] a woman unmatched in $M$ (by Step~0 and Step~2) or 
\item[(4)] a woman $M(y)$ where $y$ has an edge labeled $(+,+)$ incident on it (by Step~1 and Step~2). 
\end{itemize}

\begin{figure}[h]
	\vspace*{-5mm}
\centerline{\resizebox{0.735\textwidth}{!}{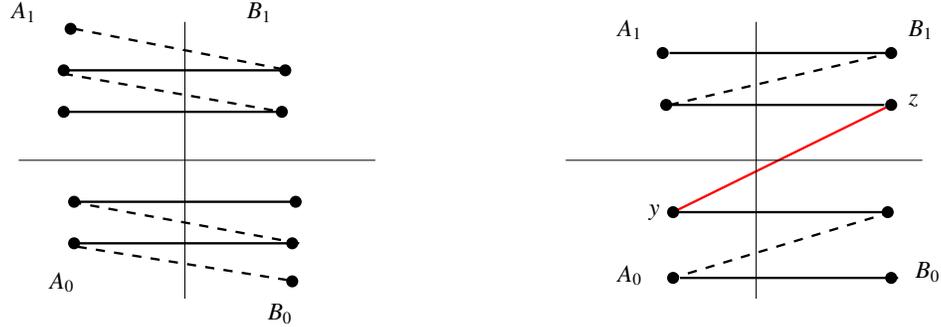}}
\caption{Vertices get added to $A_1$ and $A_0$ by alternating paths in $G_M$ from either unmatched vertices or endpoints of edges labeled~$(+,+)$. The bold black edges are in $M$ and the red edge $(y,z)$ is labeled $(+,+)$ with respect to~$M$.}
\label{fig:lem4}
\vspace*{-5mm}
\end{figure}

We show the following lemma here and its proof is based on the characterization of dominant matchings in
terms of conditions~(i)-(iv) as given by Corollary~\ref{cor0}. We will also use (1)-(4) observed above in our proof.
\begin{lemma}
\label{clm3}
$A_0 \cap A_1 = \emptyset$.
\end{lemma}
\begin{proof}
\noindent{\em Case~1.} Suppose $a$ satisfies reasons~(1) and (3) for its inclusion in $A_1$ and in $A_0$, respectively. 
So $a$ is in $A_1$ because it is reachable via an even alternating path in $G_M$ from 
an unmatched man $u$; also $a$ is in $A_0$ because it is reachable via an odd length alternating path in $G_M$ from an unmatched 
woman~$v$. Then there is an augmenting path $\langle u,\ldots,v\rangle$ wrt $M$ in $G_M$ -- a contradiction to the fact that $M$
is dominant (by Lemma~\ref{thm:domn-char}).

\smallskip

\noindent{\em Case~2.} Suppose $a$ satisfies reasons~(1) and (4) for its inclusion in $A_1$ and in $A_0$, respectively. 
So $a$ is in $A_1$ because it is reachable via an even alternating path wrt $M$ in $G_M$ from 
an unmatched man $u$; also $a$ is in $A_0$ because it is reachable via an odd length alternating path in $G_M$ from $z$, where edge 
$(y,z)$ is labeled~$(+,+)$. Then there is  an alternating path wrt $M$ in $G_M$ from an unmatched man $u$ to the edge $(y,z)$
labeled $(+,+)$ and this is a contradiction to condition~(ii) of popularity.

\smallskip

\noindent{\em Case~3.} Suppose $a$ satisfies reasons~(2) and (3) for its inclusion in $A_1$ and in $A_0$, respectively. 
This case is absolutely similar to Case~2. This will cause an alternating path wrt $M$ in $G_M$ from an unmatched woman to an edge
labeled $(+,+)$, a contradiction again to condition~(ii) of popularity.

\smallskip

\noindent{\em Case~4.} Suppose $a$ satisfies reasons~(2) and (4) for its inclusion in $A_1$ and in $A_0$, respectively. 
So $a$ is reachable via an even length alternating path in $G_M$ from an edge labeled $(+,+)$ and 
$M(a)$ is also reachable via an even length alternating path in $G_M$ from an edge  labeled~$(+,+)$.
If it is the same edge labeled $(+,+)$ that both $a$ and $M(a)$ are 
reachable from, then there is an alternating cycle in $G_M$ with a $(+,+)$ edge -- a contradiction to condition~(i) of popularity. 
If these are 2 different edges 
labeled $(+,+)$, then we have an alternating path in $G_M$ with two edges labeled $(+,+)$ -- a contradiction to condition~(iii) of popularity.

These four cases finish the proof that $A_0\cap A_1 = \emptyset$. \qed
\end{proof}

We now describe the construction of the matching $M'$. Initially $M' = \emptyset$. 
\begin{itemize}
\item For each $a \in A_0$: add the edges $(a_0,M(a))$ and $(a_1,d(a))$ to~$M'$.
\item For each $a \in A_1$: add the edge $(a_0,d(a))$ to $M'$ and if $a$ is matched in $M$ then add $(a_1,M(a))$ to~$M'$.
\item For $a \notin (A_0\cup A_1)$: add the edges $(a_0,M(a))$ and $(a_1,d(a))$ to~$M'$.

{\em (Note that the men outside $A_0 \cup A_1$ are not reachable from either unmatched vertices or edges labeled $(+,+)$ via alternating paths in~$G_M$.)}
\end{itemize}

\begin{lemma}
\label{lem:stable}
$M'$ is a stable matching in~$G'$.
\end{lemma}
\begin{proof}
Suppose $M'$ is not stable in~$G'$. Then there are edges $(u_i,v)$ and $(a_j,b)$ in $M'$ where $i,j \in\{0,1\}$,
such that in the graph $G'$, the vertices $v$ and $a_j$ prefer each other to $u_i$ and $b$, respectively. There cannot be a blocking 
pair involving a dummy woman, thus the edges $(u,v)$ and $(a,b)$ are in~$M$. 

If $i = j$, then the pair $(a,v)$ blocks $M$ in~$G$. However, from the construction of the sets $A_0,A_1,B_0,B_1$,
we know that all the blocking pairs with respect to $M$ are in $A_0 \times B_1$.
Thus there is no blocking pair in $A_0 \times B_0$ or in $A_1 \times B_1$ with respect to $M$ and so~$i \ne j$. Since $v$ prefers $a_j$ 
to $u_i$ in $G'$, the only possibility is  $i=0$ and~$j=1$. It has to be the case that $a$ prefers $v$ to $b$, 
so there is an edge labeled $(+,-)$ between $a \in A_1$ and $v \in B_0$ (see Fig.~\ref{fig:lem5}). 

\begin{figure}[h]
\centerline{\resizebox{0.3\textwidth}{!}{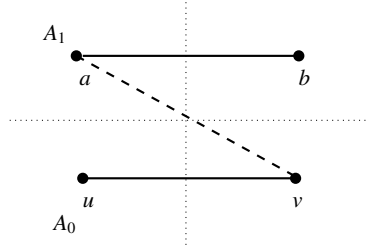}}
\caption{If the vertex $a_1$ prefers $v$ to $b$ in $G'$, then $a$ prefers $v$ to $b$ in $G$; thus the edge $(a,v)$ has to be present in~$G_M$.}
\label{fig:lem5}
\end{figure}

So once $v$ got added to $B_0$, since $a$ is adjacent in $G_M$ to a vertex in $B_0$,
vertex $a$ satisfied Step~2 of our algorithm to construct the sets $A_0,A_1,B_0$, and~$B_1$. 
So  $a$ would have got added to $A_0$ as well, i.e., $a \in A_0 \cap A_1$, a contradiction to Lemma~\ref{clm3}.
Thus there is no blocking pair with respect to $M'$ in~$G'$. \qed
\end{proof}

For each $a \in A$, note that exactly one of $(a_0,d(a))$, $(a_1,d(a))$ is in~$M'$. 
In order to form the set $T(M')$, the edges of $M'$ with women in $d(A)$ are pruned and each edge 
$(a_i,b) \in M'$, where $b \in B$ and $i \in \{0,1\}$, is replaced by~$(a,b)$. It is easy to see that $T(M') = M$. 

This concludes the proof that 
every dominant matching in $G$ can be realized as an image under $T$ of some stable matching in~$G'$. Thus $T$ is surjective.

Our mapping $T$ can also be used to solve the min-cost dominant matching problem in polynomial time. Here  we are given a cost 
function $c: E \rightarrow Q$ and the problem is to find a dominant matching in $G$ whose sum of edge costs is the least.  
We will use the mapping $T$ established from $\{$stable matchings in $G'\}$ to $\{$dominant matchings in $G\}$
to solve the min-cost dominant matching problem in~$G$. It is easy to extend $c$ to the edge set of~$G'$. For each edge $(a,b)$ in $G$,
we will assign $c(a_0,b) = c(a_1,b) = c(a,b)$ and we will set $c(a_0,d(a)) = c(a_1,d(a)) = 0$. Thus the cost of any stable matching $M'$ in $G'$ is the same is the cost of the dominant matching $T(M')$ in~$G$. 

Since every dominant matching $M$ in $G$ equals $T(M')$ for some stable matching $M'$ in $G'$, it follows 
that the min-cost dominant matching problem in $G$ is the same as the min-cost stable matching problem in~$G'$. Since a
min-cost stable matching in $G'$ can be computed in polynomial time, we can conclude Theorem~\ref{thm:min-cost-dom}. 
\begin{theorem}
\label{thm:min-cost-dom}
Given a graph $G = (A \cup B,E)$ with strict preference lists and a cost function $c: E \rightarrow \mathbb{Q}$, 
the problem of computing a min-cost dominant matching can be solved in polynomial time.
\end{theorem}

\section{The popular edge problem}
\label{sec:pop-edge}
In this section we show a decomposition for any popular matching in terms of a stable matching and a dominant matching.
We use this result to design a linear time algorithm for the {\em popular edge} problem. 
Here we are given an edge $e^*=(u,v)$ in $G = (A\cup B,E)$ and we would like to know if there exists a popular matching in $G$ that contains~$e^*$.
We claim the following algorithm solves the above problem.

\begin{enumerate}
\item Check if there is a stable matching $M_{e^*}$ in $G$ that contains edge~$e^*$. If so, then return~$M_{e^*}$.
\item Check if there is a dominant matching $M'_{e^*}$ in $G$ that contains edge~$e^*$. If so, then return~$M'_{e^*}$.
\item Return ``there is no popular matching that contains edge $e^*$ in~$G$''.
\end{enumerate}

\paragraph{Running time of the above algorithm.} In step~1 of our algorithm, we have to determine if there exists a stable matching $M_{e^*}$ in $G$ that 
contains $e^* = (u,v)$. We modify the Gale-Shapley algorithm so that the woman $v$ rejects all proposals from anyone worse than~$u$. 
If the modified Gale-Shapley algorithm produces a matching $M$ containing $e^*$, then it will be the most men-optimal stable matching 
containing $e^*$ in~$G$. Else there is no stable matching in $G$ that contains~$e^*$. We refer the reader to~\cite[Section 2.2.2]{GI89} for
the correctness of the modified Gale-Shapley algorithm; it is based on the following fact:
\begin{itemize}
\item If $G$ admits a stable matching that contains $e^*=(u,v)$, then exactly one of (1), (2), (3) occurs in any stable matching $M$ of $G$: 
{\em (1)~$e^* \in M$, (2)~$v$ is matched to a neighbor better than $u$, (3)~$u$ is matched to a neighbor better than $v$}.
\end{itemize}

In step~2 of our algorithm for the popular edge problem, we have to determine if there exists a dominant matching in $G$
that contains $e^* = (u,v)$. This is equivalent to checking if there exists a stable matching in $G'$ that contains either the edge
$(u_0,v)$ or the edge~$(u_1,v)$. This can be determined by using the same modified Gale-Shapley algorithm as given in the previous paragraph.
Thus both steps~1 and 2 of our algorithm can be implemented in $O(m)$ time. 

\paragraph{Correctness of the algorithm.}
Let $M$ be a popular matching in $G$ that contains edge~$e^*$. We will use $M$ to 
show that there is either a stable matching or a dominant matching that contains~$e^*$. 
As before, label each edge $e = (a,b)$ outside $M$ by the pair of votes $(\alpha_e,\beta_e)$, where
$\alpha_e$ is $a$'s vote for $b$ vs.\ $M(a)$ and $\beta_e$ is $b$'s vote for $a$ vs.~$M(b)$. 

We run the following algorithm now -- this is similar to the algorithm in Section~\ref{sec:onto} to build the subsets $A_0,A_1$ of $A$ and
$B_0,B_1$ of $B$, except that all the sets $A_0,A_1,B_0,B_1$ are initialized to empty sets here.
\begin{enumerate}
\item[0.] Initialize $A_0 = A_1 = B_0 =  B_1 = \emptyset$.
\item[1.] For every edge $(a,b) \in M$ that is labeled $(+,+)$: 
\begin{itemize}
\item let $A_0 = A_0 \cup \{a\}$, \ $B_1 = B_1 \cup \{b\}$, \ $A_1 = A_1 \cup \{M(b)\}$, \ and $B_0 = B_0 \cup \{M(a)\}$.
\end{itemize}
\item[2.] While there exists a man $a' \notin A_0$ that is adjacent in $G_M$ to a woman in $B_0$ do: 
\begin{itemize}
\item $A_0 = A_0 \cup \{a'\}$ and $B_0 = B_0 \cup \{M(a')\}$.
\end{itemize}
\item[3.] While there exists a woman $b' \notin B_1$ that is adjacent in $G_M$ to a man in $A_1$ do: 
\begin{itemize}
\item $B_1 = B_1 \cup \{b'\}$ and $A_1 = A_1 \cup \{M(b)\}$.
\end{itemize}
\end{enumerate}
All vertices added to the sets $A_0$ and $B_1$ are matched in $M$ -- otherwise there would be an alternating path from an
unmatched vertex to an edge labeled $(+,+)$ and this contradicts condition~(ii) of popularity of $M$ (see Theorem~\ref{thm:pop-char}).
Note that every vertex in $A_1$ is reachable via an even length alternating path wrt $M$ in $G_M$ from some man $M(b)$ 
whose partner $b$ has an edge labeled $(+,+)$ incident on it. Similarly, every vertex in $A_0$ is reachable via an odd length 
alternating path wrt $M$ in $G_M$ from some woman $M(a)$ whose partner $a$ has an edge labeled $(+,+)$ incident on it.
The proof of Case~4 of Lemma~\ref{clm3} shows that $A_0 \cap A_1 = \emptyset$.

We have $B_1 = M(A_1)$ and $B_0 = M(A_0)$ (see Fig.~\ref{fig:first}).  All edges labeled $(+,+)$ are in 
$A_0 \times B_1$ (from our algorithm) and all edges in $A_1 \times B_0$ have to be labeled $(-,-)$ (otherwise we would contradict either
condition~(i) or (iii) of popularity of $M$).

Let $A' = A_0 \cup A_1$ and $B' = B_0 \cup B_1$. Let $M_0$ be the matching $M$ restricted to $A' \cup B'$.  
The matching $M_0$ is popular on $A' \cup B'$. Suppose not and there is a matching $N_0$ on $A' \cup B'$ that is more popular.
Then the matching $N_0 \cup (M\setminus M_0)$ is more popular than $M$, a contradiction to the popularity of~$M$.
Since $M_0$ matches all vertices in $A' \cup B'$, it follows that $M_0$ is dominant on $A' \cup B'$.

\begin{figure}[h]
	\vspace*{-5mm}
\centerline{\resizebox{0.72\textwidth}{!}{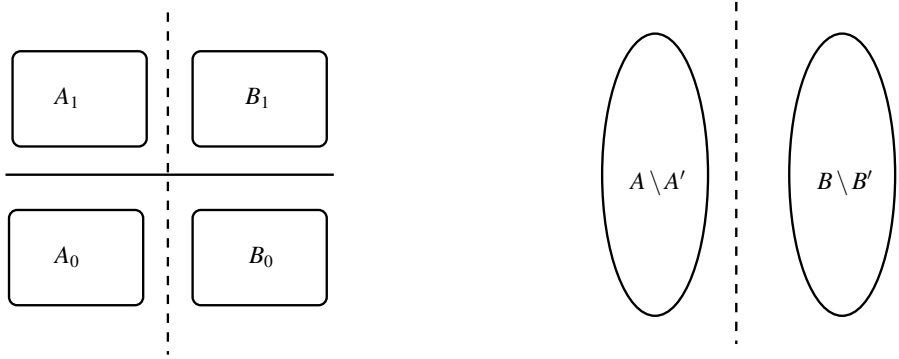}}
\caption{$M_0$ is the matching $M$ restricted to $A' \cup B'$. All unmatched vertices are in $(A \setminus A') \cup (B \setminus B')$.}
\label{fig:first}
\vspace*{-5mm}
\end{figure}
Let $M_1 = M \setminus M_0$ and let $Y = A\setminus A'$ and $Z = B \setminus B'$. 
The matching $M_1$ is stable on $Y \cup Z$ as there is no edge labeled $(+,+)$ in
$Y \times Z$ (all such edges are in $A_0 \times B_1$ by Step~1 of our algorithm above). 

The subgraph $G_M$ contains no edge in $A_1 \times Z$ --
otherwise such a woman $z \in Z$ should have been in $B_1$ (by Step~3 of the algorithm above) 
and similarly, $G_M$ contains no edge in $Y \times B_0$ -- otherwise such a man $y \in Y$ should 
have been in $A_0$ (by Step~2 of this algorithm). We will now show Lemmas~\ref{lem:domn-e} and \ref{lem:stab-e}. 
These lemmas prove the correctness of our algorithm.


\begin{lemma}
\label{lem:domn-e}
If the edge $e^* \in M_0$ then there exists a dominant matching in $G$ that contains~$e^*$.
\end{lemma}
\begin{proof}
Let $H$ be the induced subgraph of  $G$ on $Y \cup Z$. We will transform the stable matching $M_1$ in $H$ 
to a dominant matching $M^*_1$ in~$H$. We do this by computing a stable matching in the graph $H' = (Y' \cup Z', E')$ -- 
the definition of $H'$ (with respect to $H$) is analogous to the definition of $G'$ (with respect to $G$) in Section~\ref{sec:dom-mat}.
So for each man $y \in Y$, we have two men $y_0$ and $y_1$ in $Y'$ and one dummy woman $d(y)$ in $Z'$; the set $Z' = Z \cup d(Y)$ and
the preference lists of the vertices in $Y' \cup Z'$ are exactly as given in Section~\ref{sec:char} for the vertices in~$G'$.

\smallskip

We wish to compute a dominant matching in $H$, equivalently, a stable matching in~$H'$. However
we will not compute a stable matching in $H'$ from scratch since we want to obtain a dominant matching in $H$ using~$M_1$.
So we compute a stable matching in $H'$ by starting with the following matching in $H'$ (this is essentially the same as $M_1$): 
\begin{itemize}
\item for each edge $(y,z)$ in $M_1$, include the edges $(y_0,z)$ and $(y_1,d(y))$ in this initial matching and for each unmatched 
man $y$ in $M_1$, include the edge $(y_0,d(y))$ in this matching. This is a feasible starting matching as there is no blocking pair
with respect to this matching.
\end{itemize}

Now run the Gale-Shapley algorithm in $H'$ with unmatched men proposing and women disposing. Note that the starting set
of unmatched men is the set of all men $y_1$ where $y$ is unmatched in~$M_1$. However as the algorithm progresses, other men could also get unmatched
and propose. Let $M'_1$ be the resulting stable matching in~$H'$. Let $M^*_1$ be the dominant matching in $H$ corresponding to the stable matching 
$M'_1$ in~$H'$. 

Observe that $M_0$ is untouched by the transformation $M_1 \leadsto M^*_1$. 
Let $M^* = M_0\cup M^*_1$.  Since $e^* \in M_0$, the matching $M^*$ contains~$e^*$.

\begin{new-claim}
\label{clm2}
$M^*$ is a dominant matching in~$G$.
\end{new-claim}

The proof of Claim~\ref{clm2} involves some case analysis and is given at the end of this section.
Thus there is a dominant matching $M^*$ in $G$ that contains $e^*$ 
and this finishes the proof of Lemma~\ref{lem:domn-e}. \qed
\end{proof}

\begin{lemma}
\label{lem:stab-e}
If the edge $e^* \in M_1$ then there exists a stable matching in $G$ that contains~$e^*$.
\end{lemma}
\begin{proof}
Here will leave $M_1$ untouched and transform the dominant matching $M_0$ on $A' \cup B'$ to a stable matching $M'_0$ on $A' \cup B'$. We do this by 
{\em demoting} all men in~$A_1$. That is, we run the stable matching algorithm on $A' \cup B'$ with preference lists as in 
the original graph $G$, i.e., men in $A_1$ are not promoted over the ones in~$A_0$. Our starting matching is $M_0$ restricted 
to edges in~$A_1 \times B_1$. Since there is no blocking pair with respect to $M_0$ in $A_1 \times B_1$, this is a feasible starting matching.

Now unmatched men (all those in $A_0$) propose in decreasing order of preference to the women in $B'$ and when a woman receives a better proposal than 
what she currently has, she discards her current partner and accepts the new proposal. This may make men in $A_1$ single and so they too propose. This is 
Gale-Shapley algorithm with the only difference being in our starting matching not being empty but  $M_0$ restricted to the edges of~$A_1 \times B_1$. 
Let $M_0'$ be the resulting matching on~$A'\cup B'$. Let $M' = M'_0 \cup M_1$. This is a matching that contains the edge $e^*$ since $e^* \in M_1$. 
\begin{new-claim}
\label{clm4}
$M'$ is a stable matching in~$G$.
\end{new-claim}

The proof of Claim~\ref{clm4} again involves some case analysis and is given at the end of this section.
Thus there is a stable matching $M'$ in $G$ that contains 
the edge $e^*$ and this finishes the proof of Lemma~\ref{lem:stab-e}. \qed
\end{proof}

We have thus shown the correctness of our algorithm. We  can now conclude the following theorem.

\begin{theorem}
\label{thm:popular-edge}
Given a stable marriage instance $G = (A \cup B,E)$ with strict preference lists and an edge $e^* \in E$, we can determine in
linear time if there exists a popular matching in $G$ that contains~$e^*$.
\end{theorem}

\subsubsection{Proof of Claim~\ref{clm2}.}
We need to show that $M^* = M_0 \cup M^*_1$ is a dominant matching,
where $M^*_1$ is the dominant matching in $H$ corresponding to the stable matching $M'_1$ in~$H'$.

Let $Y_0$ be the set of men $y \in Y$ such that $(y_1,d(y)) \in M'_1$
and let $Y_1$ be the set of men $y \in Y$ such that $(y_0,d(y)) \in M'_1$.
Let $Z_1$ be the set of those women in $Z$ that are matched in $M'_1$ to men in $Y_1$ and let $Z_0 = Z\setminus Z_1$.

The following properties will be useful to us:
\begin{itemize}
\item[(i)] If $y \in Y_1$, then $M^*_1(y)$ ranks at least as good as $M_1(y)$ in $y$'s preference list. This is because $y \in Y_1$ and note that $Y_1$ is a 
{\em promoted} set when compared to~$Y_0$. Thus $y_1$ gets at least as good a partner in $M^*_1$ as in the men-optimal stable matching in $H$, 
which is at least as good as $M_1(y)$, as $M_1$ is a stable matching in~$H$.
\item[(ii)] If $z \in Z_0$, then $M^*_1(z)$ ranks at least as good as $M_1(z)$ in $z$'s preference list. This is because in the computation of the 
stable matching $M'_1$, if the vertex $z$ rejects $M_1(z)$, then it was upon receiving a better proposal from a neighbor in $Y_0$ (since $z \in Z_0$). 
Thus $z$'s final partner in $M'_1$, and hence in $M^*_1$, ranks at least as good as $M_1(z)$ in her preference list.
\end{itemize}

Label each edge $e=(a,b)$ in $E \setminus M^*$ by the pair of votes $(\alpha_e,\beta_e)$, where
$\alpha_e$ is $a$'s vote for $b$ vs.\ $M^*(a)$ and $\beta_e$ is $b$'s vote for $a$ vs.~$M^*(b)$.
We will first show the following claim here. 
\begin{claim}
Every edge in $(A_1 \cup Y_1) \times (B_0 \cup Z_0)$ is labeled $(-,-)$ with respect to~$M^*$.
\end{claim}
We already know that every edge in $A_1 \times B_0$ is labeled $(-,-)$ with respect to $M_0$ and
as shown in part~(1) of Claim~\ref{clm1}, it is easy to see that every edge in $Y_1 \times Z_0$ is labeled $(-,-)$ with respect to~$M^*_1$.
We will now show that all edges in $(Y_1 \times B_0) \cup (A_1 \times Z_0)$ are labeled~$(-,-)$ with respect to $M$.
\begin{itemize}
\item Consider any edge $(y,b) \in Y_1 \times B_0$. We know that $(y,b)$ was labeled $(-,-)$ with respect to~$M$.  We have 
$M^*(b) = M_0(b) = M(b)$. Thus $b$ prefers $M^*(b)$ to~$y$. 
The man $y$ preferred $M(y)$ to $b$ and since $y \in Y_1$, we know from (i) above that $y$ ranks $M^*_1(y)$ at least as good as
$M_1(y) = M(y)$. Thus the edge $(y,b)$ is labeled $(-,-)$ with respect to $M^*$ as well.

\item Consider any edge in $(a,z) \in A_1 \times Z_0$. We know that $(a,z)$ was labeled $(-,-)$ with respect to~$M$.
We have $M^*(a) = M_0(a) = M(a)$. Thus $a$ prefers $M^*(a)$ to~$z$. The woman $z$ preferred $M_1(z)$ to $a$ and we know from
(ii) above that $z$ ranks $M^*_1(z)$ at least as good as~$M_1(z)$.  Thus the edge $(a,z)$ is labeled $(-,-)$ 
with respect to $M^*$ as well.
\end{itemize}
Thus we have shown that every edge in $(A_1 \cup Y_1) \times (B_0 \cup Z_0)$ is labeled~$(-,-)$. 
We will now show the following claim.
\begin{claim}
Any edge labeled $(+,+)$ with respect to $M^*$ has to be in $(A_0 \cup Y_0) \times (B_1 \cup Z_1)$. 
\end{claim}
\begin{figure}[h]
\centerline{\resizebox{0.4\textwidth}{!}{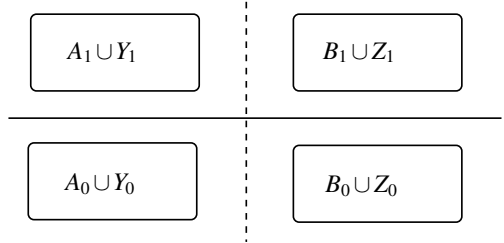}}
\caption{All edges in $(A_1\cup Y_1)\times(B_0\cup Z_0)$ are labeled $(-,-)$ wrt $M^*$ and all edges labeled $(+,+)$ wrt $M^*$ are in $(A_0\cup Y_0) \times (B_1\cup Z_1)$.}
\label{fig:lemma6}
\end{figure}

Note that we already know that no edge in $A_i \times B_i$ 
is labeled $(+,+)$ wrt $M_0$ and no edge in $Y_i \times Z_i$ is labeled $(+,+)$ wrt $M^*_1$, for $i = 0,1$. 
We will now show that no edge in $\cup_{i=0}^1(A_i \times Z_i) \cup (Y_i \times B_i)$ is labeled $(+,+)$ (see Fig.~\ref{fig:lemma6}).
\begin{enumerate}
\item[(1)] Consider any edge in $(a,z) \in A_1 \times Z_1$. We know that $(a,z)$ was labeled $(-,-)$ wrt $M$. Since $M^*(a) = M_0(a) = M(a)$, 
the first coordinate in this edge label wrt $M^*$ is still~$-$. Thus this edge is not labeled $(+,+)$ wrt~$M^*$.

\item[(2)] Consider any edge in $(y,b) \in Y_1 \times B_1$: there was no edge labeled $(+,+)$ wrt $M$ in~$Y \times B_1$. 

-- Suppose $(y,b)$ was labeled $(-,-)$ or $(+,-)$ wrt~$M$. Since $M^*(b) = M_0(b) = M(b)$, the second coordinate in this edge label with 
respect to $M^*$ is still~$-$. Thus this edge is not labeled $(+,+)$ wrt~$M^*$. 

-- Suppose $(y,b)$ was labeled $(-,+)$ wrt~$M$. Since $y \in Y_1$,  we know from (i) above that $y$ ranks $M^*_1(y)$ at least as good as~$M_1(y)$. 
Hence the first coordinate in this edge label wrt $M^*$ is still~$-$. Thus this edge is not labeled $(+,+)$ wrt~$M^*$.

\item[(3)] Consider any edge $(y,b) \in Y_0 \times B_0$: we know that $(y,b)$ was labeled $(-,-)$ wrt~$M$. Since $M^*(b) = M_0(b) = M(b)$, 
the second coordinate in this edge label wrt $M^*$ is still~$-$. Thus this edge is not labeled $(+,+)$ wrt~$M^*$.

\item[(4)] Consider any edge in $(a,z) \in A_0 \times Z_0$: there was no edge labeled $(+,+)$ wrt $M$ in~$A_0 \times Z$. 

-- Suppose $(a,z)$ was labeled $(-,-)$ or $(-,+)$ wrt~$M$. Since $M^*(a) = M_0(a) = M(a)$, the first coordinate in this edge label wrt $M^*$ is 
still~$-$. Thus this edge is not labeled $(+,+)$ wrt~$M^*$. 

-- Suppose $(a,z)$ was labeled $(+,-)$ wrt~$M$. Since $z \in Z_0$, we know from (ii) above that $z$ ranks $M^*_1(z)$ at least as good as~$M_1(z)$. 
Hence the second coordinate in this edge label wrt $M^*$ is still~$-$. Thus this edge is not labeled $(+,+)$  wrt~$M^*$.
\end{enumerate}

Thus any edge labeled $(+,+)$ has to be in $(A_0 \cup Y_0) \times (B_1 \cup Z_1)$.
This fact along with the earlier claim that all edges in $(A_1 \cup Y_1) \times (B_0 \cup Z_0)$ are labeled $(-,-)$, immediately implies that 
Claim~\ref{clm1} holds here, where
we assign $f$-values to all vertices in $A \cup B$ as follows: if $a \in A_1 \cup Y_1$ then $f(a) = 1$ else $f(a) = 0$; 
similarly, if $b \in B_1 \cup Z_1$ then $f(b) = 1$ else~$f(b) = 0$. 

Thus if the edge $(a,b)$ is labeled $(+,+)$, then $f(a) = 0$ and $f(b) = 1$, and
if $(y,z)$ is an edge such that $f(y) = 1$ and $f(z) = 0$, then $(y,z)$ has to be labeled~$(-,-)$.
Lemmas~\ref{lem:aug-path} and \ref{lem:popular} with $M^*$ replacing $M$ follow now (since all they need is Claim~\ref{clm1}). 
We can conclude that $M^*$ is dominant in~$G$. Thus there is a dominant matching in $G$ that contains~$e^*$. \qed

\subsubsection{Proof of Claim~\ref{clm4}.}
We will now show that $M' = M'_0 \cup M_1$ is a stable matching. We already know that there is no edge labeled $(+,+)$ in $A' \times B'$ with respect 
to $M'_0$ and there is no edge labeled $(+,+)$ in $Y \times Z$ with respect to~$M_1$. Now we need to show that there is no edge labeled $(+,+)$ 
either in $A' \times Z$ or in~$Y \times B'$. 

Label each edge $e=(a,b)$ in $E\setminus M'$ by the pair of 
votes $(\alpha_e,\beta_e)$, where $\alpha_e$ is $a$'s vote for $b$ vs.\ $M'(a)$ and $\beta_e$ is $b$'s vote for $a$ vs.~$M'(b)$.
We will first show that there is no edge labeled $(+,+)$ in $A' \times Z$, i.e., in $(A_0\cup A_1) \times Z$.

\begin{itemize}
\item[(1)] Consider any $(a,z) \in A_1 \times Z$: this edge was labeled $(-,-)$ wrt~$M$. Since $M'(z) = M_1(z) = M(z)$,
the second coordinate of the label of this edge wrt $M'$ is~$-$. Thus this edge cannot be labeled $(+,+)$ wrt~$M'$. 

\item[(2)] Consider any $(a,z) \in A_0 \times Z$: there was no edge labeled $(+,+)$ wrt $M$ in $A' \times Z$.  

-- Suppose $(a,z)$ was labeled $(+,-)$ or $(-,-)$ wrt~$M$. Since $M'(z) = M_1(z) = M(z)$,
the second coordinate of the label of this edge wrt $M'$ is~$-$.

-- Suppose $(a,z)$ was labeled~$(-,+)$. Since $a \in A_0$, his neighbor $M'_0(a)$ is ranked at least as good as $M_0(a)$ in his preference
list. This is because 
women in $B_0$ are unmatched in our starting matching and no woman $b \in B_0$ prefers any neighbor in $A_1$ to $M_0(b)$ (all edges in
$A_1\times B_0$ are labeled $(-,-)$ wrt $M_0$). Thus in our algorithm that computes~$M'_0$, $a$ will get accepted either by $M_0(a)$ or a better neighbor.
Hence the first coordinate of this edge label wrt $M'$ is still~$-$.
\end{itemize}

We will now show that there is no edge labeled $(+,+)$  with respect to $M'$ in  $Y \times B'$, i.e., in $Y \times (B_0 \cup B_1)$.
\begin{itemize}
\item[(3)] Consider any $(y,b) \in Y \times B_0$: the edge $(y,b)$ was labeled $(-,-)$ wrt~$M$. Since $M'(y) = M_1(y) = M(y)$,
the first coordinate of the label of this edge wrt $M'$ is~$-$. Thus this edge cannot be labeled $(+,+)$ wrt~$M'$. 
\item[(4)] Consider any $(y,b) \in Y \times B_1$: there was no edge labeled $(+,+)$ wrt $M$ in $Y \times B'$. 

-- Suppose $(y,b)$ was labeled $(-,+)$ or $(-,-)$ wrt~$M$. Since $M'(y) = M_1(y) = M(y)$,
the first coordinate of the label of this edge wrt $M'$ is~$-$.

-- Suppose $(y,b)$ was labeled~$(+,-)$. Since $b \in B_1$, her neighbor $M'_0(b)$ is ranked at least as good as $M_0(b)$ in her preference
list. This is because our starting matching matched $b$ to $M_0(b)$ and $b$ would reject $M_0(b)$ only upon receiving a better proposal. Thus 
the second coordinate of the label of this edge with respect to $M'$ is~$-$.
\end{itemize}

This completes the proof that there is no edge labeled $(+,+)$ with respect to $M'$ in~$G$. In other words, $M'$ is a stable matching in~$G$. \qed

\section{Finding an unstable popular matching}
\label{sec:dom-vs-stab}
We are given $G = (A \cup B,E)$ with strict preference lists and we would like to know if every popular matching in $G$ is also stable.
In order to answer this question, we could compute a dominant matching $D$ and a stable matching $S$ in~$G$. If $|D| > |S|$, then
it is obviously the case that not every popular matching in $G$ is stable. However it could be the case that $D$ is stable (and so $|D| = |S|$).

We now show an efficient algorithm to check if $\{$popular matchings$\} = \{$stable matchings$\}$ or not in $G$. Note that in the latter
case, we have $\{$stable matchings$\} \subsetneq \{$popular matchings$\}$ in~$G$.

\smallskip

Let $G$ admit an unstable popular matching~$M$. We know that $M$ can be partitioned into $M_0 \cupdot M_1$, as described in Section~\ref{sec:pop-edge}.
Here $M_0$ is a dominant matching on $A' \cup B'$ and $M_1$ is stable on $Y \cup Z$, where $Y = A \setminus A'$ and $Z = B \setminus B'$ (refer to Fig.~\ref{fig:first}).
Since $M$ is unstable, there is an edge $(a,b)$ that blocks~$M$. Since there is no blocking pair involving on any vertex in $Y \cup Z$,
it has to be the case that $a \in A'$ and $b \in B'$ (in particular, $a \in A_0$ and $b \in B_1$).

Run the transformation $M_1 \leadsto M^*_1$ performed in the proof of Lemma~\ref{lem:domn-e}. Claim~\ref{clm2} tells us that 
$M^* = M_0 \cup M^*_1$ is a dominant matching. The edge $(a,b)$ is a blocking pair to $M^*$ since $M^*(a) = M_0(a)$ and $M^*(b) = M_0(b)$, 
so $a$ and $b$  prefer each other to their respective partners in~$M^*$. Thus $M^*$ is an unstable dominant matching and
Lemma~\ref{non-stab-domn} follows.

\begin{lemma}
\label{non-stab-domn}
If $G$ admits an unstable popular matching then $G$ admits an unstable dominant matching.
\end{lemma}

Hence in order to answer the question of whether every popular matching in $G$ is stable or not, we need to decide if there
exists a dominant matching $M$ in $G$ with a blocking pair.  
We will use the mapping $T: \{$stable matchings in $G'\} \rightarrow \{$dominant matchings in $G\}$ defined in Section~\ref{sec:dom-mat} here. 
Our task is to determine if there exists a stable matching in $G'$ that includes a pair of edges $(a_0,v)$ and $(u_1,b)$ such that $a$ and $b$ 
prefer each other to $v$ and $u$, respectively, in~$G$. It is easy to decide in $O(m^3)$ time whether such a stable matching exists or not in~$G'$. 
\begin{itemize}
\item For every pair of edges $e_1 = (a,v)$ and $e_2 = (u,b)$ in $G$ such that
$a$ and $b$ prefer each other to $v$ and $u$, respectively: determine if there is a stable matching in $G'$ that contains the pair of edges
$(a_0,v)$ and~$(u_1,b)$. 
\end{itemize}

An algorithm to construct a stable matching that contains a pair of edges (if such a matching exists) is similar to the algorithm described earlier 
to construct a stable matching that contains a single edge~$e$. In the graph $G'$, we modify Gale-Shapley algorithm so that $b$ rejects proposals 
from all neighbors ranked worse than $u_1$ and $v$ rejects all proposals from neighbors ranked worse than~$a_0$. 

If the resulting algorithm returns a stable matching that contains the edges $(a_0,v)$ and $(u_1,b)$, then we have the desired matching; else $G'$ has no stable 
matching that contains this particular pair of edges.  In order to determine if there exists an unstable dominant matching, we may need to go through all
pairs of edges $(e_1,e_2) \in E\times E$. Since we can determine in linear time if there exists a stable matching in $G'$ with any particular pair of edges,
the entire running time of this algorithm is $O(m^3)$, where $m = |E|$.

\paragraph{A faster algorithm.} It is easy to improve the running time to $O(m^2)$.
For each edge $(a,b) \in E$ we check for the following:
\begin{itemize}
\item[($\ast$)] a stable matching in $G'$ such that (1)~$a_0$ is matched to a neighbor that is ranked 
worse than $b$, and (2)~$b$ is matched to a neighbor $u_1$ where $u$ is ranked worse than $a$ in $b$'s list.
\end{itemize}

We modify the Gale-Shapley algorithm in $G'$ so that (1)~$b$ rejects all offers from level~0 neighbors, i.e., $b$ accepts proposals only from level~1 neighbors,
and (2)~every neighbor of $a_0$ that is ranked better than $b$ rejects proposals from~$a_0$. 

Suppose ($\ast$) holds.
Then this modified Gale-Shapley algorithm returns among all such stable matchings, the most men-optimal and women-pessimal one~\cite{GI89}. 
Thus among all stable matchings that match $a_0$ to
a neighbor ranked worse than $b$ and the woman $b$ to a level~1 neighbor, the matching returned by the above algorithm 
matches $b$ to its least preferred neighbor and $a_0$ to its most preferred neighbor.

Hence if the modified Gale-Shapley algorithm returns a matching that is (i)~unstable or (ii)~matches  $a_0$ to $d(a_0)$ or (iii)~matches $b$ to a 
neighbor better than $a_1$, then there is no dominant matching $M$ in $G$ such that the pair $(a,b)$ blocks~$M$. 
Else we have the desired stable matching in $G'$, call this matching~$M'$. 

The matching $T(M')$ will be a dominant matching in $G$ where the pair $(a,b)$ blocks~$T(M')$. Since we may need to go through all edges in $E$ and 
the time taken for any edge $(a,b)$ is $O(m)$, the entire running time of this algorithm is~$O(m^2)$. 
We have thus shown the following theorem.

\begin{theorem}
\label{thm:unstable}
Given $G = (A \cup B,E)$ where every vertex has a strict ranking over its neighbors, we can decide in $O(m^2)$ time whether every 
popular matching in $G$ is stable or not; if not, we can return an unstable popular matching.
\end{theorem}

\subsubsection{Conclusions and Open problems.} We considered the popular edge problem in a stable marriage instance $G= (A \cup B,E)$
with strict preference lists and showed a linear time algorithm for this problem. A natural extension is that we are given $k$ edges 
$e_1,\ldots,e_k$, for $k \ge 2$, and we would like to know if there exists a popular matching that contains {\em all} these $k$ edges. 
Another open problem is to efficiently find among all popular matchings that contain a particular edge $e^*$, one of largest size.
There are no polynomial time algorithms currently known to construct a popular matching in $G$ that is neither stable nor dominant.

The first polynomial time algorithm for the stable matching problem in general graphs (not necessarily bipartite) was given by Irving~\cite{Irv85}. 
On the other hand, the complexity of the popular matching problem in general graphs is open. 
Is there a polynomial time algorithm for the dominant matching problem in $G$? 

\medskip

\noindent{\em Acknowledgment.} Thanks to Chien-Chung Huang for useful discussions which led to the definition of dominant matchings.

\bibliographystyle{abbrv}
\bibliography{mybib}

\subsubsection*{Appendix: An overview of the maximum size popular matching algorithms in \cite{HK13,Kav12-journal}.} 
Theorem~\ref{thm:pop-char} (from \cite{HK13}) stated in Section~\ref{sec:char} showed conditions~(i)-(iii) as 
necessary and sufficient conditions for a matching to be popular in~$G = (A \cup B,E)$. It was also observed in \cite{HK13} 
that condition~(iv): {\em there is no augmenting path with respect to $M$ in $G_M$} was  
a sufficient condition for a popular matching to be a maximum size popular matching.

The goal was to construct a matching $M$ that satisfied conditions~(i)-(iv). The algorithm in \cite{HK13} computed 
appropriate subsets $L$ and $R$ of $A \cup B$ and showed that running Gale-Shapley algorithm with vertices of $L$
proposing and vertices of $R$ disposing resulted in a matching that obeyed conditions~(i)-(iv). Constructing these sets
$L$ and $R$ took $O(n_0)$ iterations, where $n_0 = \min(|A|,|B|)$; each iteration involved two invocations of the Gale-Shapley
algorithm. Thus the running time of this algorithm was~$O(mn_0)$.

A simpler and more efficient algorithm for constructing a matching that satisfied  conditions~(i)-(iv) was given in \cite{Kav12-journal}.
This algorithm worked with a graph $\tilde{G} = (\tilde{A}\cup B, \tilde{E})$
which is quite similar to the graph $G'$ used in Section~\ref{sec:dom-mat}, except that there were no dummy women in~$\tilde{G}$. 
The set $\tilde{A}$ had two copies $a_0$ and $a_1$ of each $a \in A$ and at most one of $a_0,a_1$ was {\em active} at any point in time. 
Every $b \in B$ preferred subscript~1 neighbors to subscript~0 neighbors (within subscript~$i$ neighbors, it was $b$'s original order of preference).

The algorithm here was ``active men propose and women dispose''. To begin with, only the men in $\{a_0: a \in A\}$ were active but
when a man $a_0$ got rejected by all his neighbors, he became inactive and his counterpart $a_1$ became active.  It was shown that the 
matching returned satisfied conditions~(i)-(iv). This was a linear time algorithm for computing a maximum size popular matching in $G = (A \cup B,E)$. 
\end{document}